\documentclass[11pt]{myclass}
\usepackage{euscript}

\usepackage{amssymb}
\usepackage{epsfig}
\usepackage{xspace}
\usepackage{color}

\newcommand{\denselist}{\vspace{-5pt} \itemsep -2pt\parsep=-1pt\partopsep -2pt}

\newcommand{\eps}{\varepsilon}

\newcommand{\vol}{\textsf{vol}}
\newcommand{\rad}{\textsf{rad}}

\DeclareMathOperator{\kde}{\textsc{kde}}
\DeclareMathOperator{\poly}{\textsf{poly}}

\newcommand{\prob}[1]{\ensuremath{\textbf{{\sffamily Pr}}\hspace{-.8mm}\left[#1\right]}}

\newcommand{\Reals}{\ensuremath{\mathbb{R}}}
\newcommand{\Eu}[1]{\ensuremath{\EuScript{#1}}}

\usepackage{picins}

\newlength{\ppicwd}

\title{$\eps$-Samples for Kernels}
\author{Jeff M. Phillips \\ University of Utah \\ \texttt{jeffp@cs.utah.edu}}

\begin{document}
\begin{titlepage}
\maketitle

\begin{center} \today \end{center}

\begin{abstract}
We study the worst case error of kernel density estimates via subset approximation.  A kernel density estimate of a distribution is the convolution of that distribution with a fixed kernel (e.g. Gaussian kernel).  Given a subset (i.e. a point set) of the input distribution, we can compare the kernel density estimates of the input distribution with that of the subset and bound the worst case error.  If the maximum error is $\eps$, then this subset can be thought of as an $\eps$-sample (aka an $\eps$-approximation) of the range space defined with the input distribution as the ground set and the fixed kernel representing the family of ranges.  Interestingly, in this case the ranges are not binary, but have a continuous range (for simplicity we focus on kernels with range of $[0,1]$); these allow for smoother notions of range spaces.  

It turns out, the use of this smoother family of range spaces has an added benefit of greatly decreasing the size required for $\eps$-samples.  For instance, in the plane the size is $O((1/\eps^{4/3}) \log^{2/3}(1/\eps))$ for disks (based on VC-dimension arguments) but is only $O((1/\eps) \sqrt{\log (1/\eps)})$ for Gaussian kernels and for kernels with bounded slope that only affect a bounded domain.  
These bounds are accomplished by studying the discrepancy of these ``kernel'' range spaces, and here the improvement in bounds are even more pronounced.  In the plane, we show the discrepancy is $O(\sqrt{\log n})$ for these kernels, whereas for balls there is a lower bound of $\Omega(n^{1/4})$.  
\end{abstract}
\end{titlepage}

\section{Introduction}

We study the $L_\infty$ error in kernel density estimates of points sets by a kernel density estimate of their subset.  Formally, we start with a size $n$ point set $P \subset \Reals^d$ and a kernel $K : \Reals^d \times \Reals^d \to \Reals$.  Then a kernel density estimate $\kde_P$ of a point set $P$ is a convolution of that point set with a kernel, defined at any point $x \in \Reals^d$:
\[
 \kde_P(x) = \sum_{p \in P} \frac{K(x,p)}{|P|}.
\]
%
The goal is to construct a subset $S \subset P$, and bound its size, so that it has $\eps$-bounded $L_\infty$ error, i.e. 
\[
L_\infty\left(\kde_P,\kde_S\right) = \max_{x \in \Reals^d} \left| \kde_P(x) - \kde_S(x) \right| \leq \eps.
\]
We call such a subset $S$ an \emph{$\eps$-sample of a kernel range space $(P,\Eu{K})$}, where $\Eu{K}$ is the set of all functions $K(x,\cdot)$ represented by a fixed kernel $K$ and an arbitrary center point $x \in \Reals^d$.
Our main result is the construction in $\Reals^2$ of an $\eps$-sample of size $O((1/\eps) \sqrt{\log (1/\eps)})$ for a broad variety of kernel range spaces.  

We will study this result through the perspective of three types of kernels.  We use as examples the ball kernels $\Eu{B}$, the triangle kernels $\Eu{T}$, and the Gaussian kernels $\Eu{G}$; we normalize all kernels so $K(p,p) =1$.
\makebox[\linewidth][r]{\includegraphics[width=.45\linewidth]{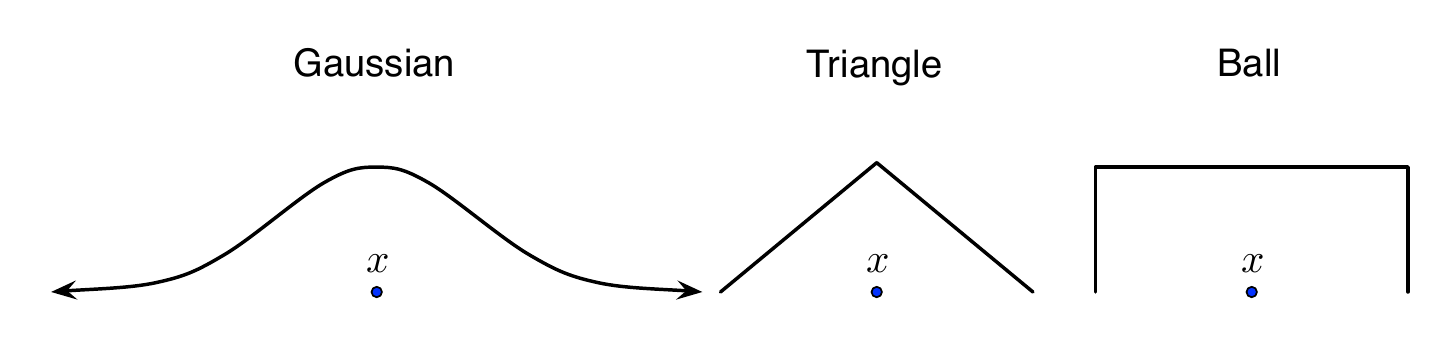}}
\vspace{-.93in}
\begin{itemize}\denselist
\item For $K \in \Eu{B}$ : $K(x,p) = \{1$ if  $\|x-p\| \leq 1$ and $0$ otherwise\}.
\item For $K \in \Eu{T}$ : $K(x,p) = \max \{0, 1- \|x-p\|\}.$
\item For $K \in \Eu{G}$ : $K(x,p) = \exp(-\|x-p\|^2).$
\end{itemize}
\vspace{-5pt}
Our main result holds for $\Eu{T}$ and $\Eu{G}$, but not $\Eu{B}$.  However, in the context of combinatorial geometry, kernels related to $\Eu{B}$ (binary ranges) seem to have been studied the most heavily from an $L_\infty$ error perspective, and require larger $\eps$-samples.  In Appendix \ref{sec:lb} we show a lower-bound that such a result cannot hold for $\Eu{B}$.  

We re-describe this result next by adapting (binary) range spaces and discrepancy; these same notions will be used to prove our result.    

\paragraph{Range spaces.}
A kernel range space is an extension of the combinatorial concept of a range space.  
Let $P \subset \Reals^d$ be a set of $n$ points.  Let $\Eu{A} \subset 2^P$ be the set of subsets of $P$, for instance when $\Eu{A} = \Eu{B}$ they are defined by containment in a ball. The pair $(P,\Eu{A})$ is called a \emph{range space}.  

Thus we can re-imagine a kernel range space $(P,\Eu{K})$ as the family of \emph{fractional} subsets of $P$, that is, each $p \in P$ does not need to be completely in ($1$) or not in ($0$) a range, but can be fractionally in a range described by a value in $[0,1]$.  In the case of the ball kernel $K(x,\cdot) \in \Eu{B}$ we say the associated range space is a \emph{binary range space} since all points have a binary value associated with each range, corresponding with \emph{in} or \emph{not in}.  

\paragraph{Colorings and discrepancy.}
Let $\chi : P \to \{-1,+1\}$ be a \emph{coloring} of $P$.  The \emph{combinatorial discrepancy} of $(P,\Eu{A})$, given a coloring $\chi$ is defined $d_\chi(P,\Eu{A}) = \max_{R \in \Eu{A}} |\sum_{p \in R} \chi(p)|$.  
For a kernel range space $(P,\Eu{K})$, this is generalized as the \emph{kernel discrepancy}, defined $d_\chi(P,\Eu{K}) = \max_{x \in \Reals^d} \sum_{p \in P} \chi(p) K(x,p)$; we can also write $d_\chi(P,K_x) = \sum_{p \in P} \chi(p) K(x,p)$ for a specific kernel $K_x$, often the subscript $x$ is dropped when it is apparent.  
Then the \emph{minimum kernel discrepancy} of a kernel range space is defined $d(P,\Eu{K}) = \min_\chi d_\chi(P,\Eu{K})$.   
See Matou\'{s}ek's~\cite{Mat99} and Chazelle's~\cite{Cha01} books for a masterful treatments of this field when restricted to combinatorial discrepancy.

\paragraph{Constructing $\eps$-samples.}
Given a (binary) range space $(P,\Eu{A})$ an \emph{$\eps$-sample} (a.k.a.  an $\eps$-approximation) is a subset $S \subset P$ such that the density of $P$ is approximated with respect to $\Eu{A}$ so
\[
\max_{R \in \Eu{A}} \left| \frac{|R \cap P|}{|P|} - \frac{|R \cap S|}{|S|} \right| \leq \eps.
\]
Clearly, an $\eps$-sample of a kernel range space is a direct generalization of the above defined $\eps$-sample for (binary) range space.  In fact, recently Joshi \etal~\cite{JKPV11} showed that for any kernel range space $(P,\Eu{K})$ where all super-level sets of kernels are described by elements of a binary range space $(P,\Eu{A})$, then an $\eps$-sample of $(P,\Eu{A})$ is also an $\eps$-sample of $(P,\Eu{K})$.  For instance, super-level sets of $\Eu{G}, \Eu{T}$ are balls in $\Eu{B}$.  

$\eps$-Samples are a very common and powerful coreset for approximating $P$; the set $S$ can be used as proxy for $P$ in many diverse applications (c.f. \cite{AAB08,Phi08,CCD11,RACZU11}).  
For binary range spaces with constant VC-dimension~\cite{VC71} a random sample $S$ of size $O((1/\eps^2) \log (1/\delta))$ provides an $\eps$-sample with probability at least $1-\delta$~\cite{LLS01}.  Better bounds can be achieved through deterministic approaches as outlined by Chazelle and Matousek~\cite{CM96}, or see either of their books for more details~\cite{Cha01,Mat99}.  This approach is based on the following rough idea.  Construct a low discrepancy coloring $\chi$ of $P$, and remove all points $p \in P$ such that $\chi(p) = -1$.  Then repeat these color-remove steps until only a small number of points are left (that are always colored $+1$) and not too much error has accrued.  
As such, the best bounds for the size of $\eps$-samples are tied directly to discrepancy.  As spelled out explicitly by Phillips~\cite{Phi08} (see also \cite{Mat99,Cha01} for more classic references), for a range space $(P,\Eu{A})$ with discrepancy $O(\log^\tau |P|)$  (resp. $O(|P|^\psi \log^\tau |P|)$) that can be constructed in time $O(|P|^w \log^\phi(|P|))$, there is an $\eps$-sample of size $g(\eps) = O((1/\eps) \log^\tau (1/\eps))$ (resp. $O(((1/\eps)\log^\tau(1/\eps))^{1/(1-\psi)})$) that can be constructed in time $O(w^{w-1} n \cdot (g(\eps))^{w-1} \cdot \log^\phi(g(\eps)) + g(\eps))$.  
Although, originally intended for binary range spaces, these results hold directly for kernel range spaces.

\subsection{Our Results}
Our main structural result is an algorithm for constructing a low-discrepancy coloring $\chi$ of a kernel range space.  The algorithm is relatively simple; we construct a min-cost matching of the points (minimizes sum of distances), and for each pair of points in the matching we color one point $+1$ and the other $-1$ at random.  
\begin{theorem}
For $P \subset \Reals^d$ of size $n$, the above coloring $\chi$, has discrepancy 
$d_\chi(P,\Eu{T}) = O(n^{1/2-1/d}\sqrt{\log (n/\delta)})$ and 
$d_\chi(P,\Eu{G}) = O(n^{1/2-1/d}\sqrt{\log (n/\delta)})$ with probability at least $1-\delta$.  
\label{thm:disc-main}
\end{theorem}

This implies an efficient algorithm for constructing small $\eps$-samples of kernel range spaces.  
\begin{theorem}
For $P \subset \Reals^d$, with probability at least $1-\delta$, we can construct in $O(n/\eps^2)$ time an $\eps$-sample of $(P,\Eu{T})$ or $(P,\Eu{G})$ of size $O((1/\eps)^{2d/(d+2)}  \log^{d/(d+2)}(1/\eps \delta))$.
\label{thm:d-main}
\end{theorem}
Note that in $\Reals^2$, the size is $O((1/\eps) \sqrt{\log(1/\eps \delta)})$, near-linear in $1/\eps$, and the runtime can be reduced to $O((n/\sqrt{\eps}) \log^5(1/\eps))$.  Furthermore, for $\Eu{B}$, the best known upper bounds for discrepancy (which are tight up to a $\log$ factor, see Appendix \ref{sec:lb}) are noticeably larger at $O((1/\eps)^{2d/(d+1)} \cdot \log^{d/(d+1)}(1/\eps))$, especially in $\Reals^2$ at $O((1/\eps)^{4/3} \log^{2/3}(1/\eps))$.  

We note that many combinatorial discrepancy results also use a matching where for each pair one is colored $+1$ and the other $-1$.  However, these matchings are ``with low crossing number''  and are not easy to construct.  For a long time these results were existential, relying on the pigeonhole principle.  But, recently Bansal~\cite{Ban10} provided a randomized constructive algorithm; also see a similar, simpler and more explicit, approach recently on the arXiv~\cite{LM12}.  Yet still these are quite more involved than our min-cost matching.  We believe that this simpler, and perhaps more natural, min-cost matching algorithm may be of independent practical interest.

\paragraph{Proof overview.}
The proof of Theorem \ref{thm:d-main} follows from Theorem \ref{thm:disc-main}, the above stated results in \cite{Phi08}, and Edmond's $O(n^3)$ time algorithm for min-cost matching $M$~\cite{Edm65}.  So the main difficulty is proving Theorem \ref{thm:disc-main}.  
We first outline this proof in $\Reals^2$ on $\Eu{T}$.  In particular, we focus on showing that the coloring $\chi$ derived from $M$, for any single kernel $K$, has $d_\chi(P,K) = O(\sqrt{\log (1/\delta)})$ with probability at least $1-\delta$.  Then in Section \ref{sec:2fam} we extend this to an entire family of kernels with $d_\chi(P,\Eu{K}) = O(\sqrt{\log(n/\delta)})$.  

The key aspect of kernels required for the proof is a bound on their slope, and this is the critical difference between binary range spaces (e.g. $\Eu{B}$) and kernel range spaces (e.g. $\Eu{T}$).  On the boundary of a binary range the slope is infinite, and thus small perturbations of $P$ can lead to large changes in the contents of a range; all lower bounds for geometric range spaces seem to be inherently based on creating a lot of action along the boundaries of ranges.  
For a kernel $K(x,\cdot)$ with slope bounded by $\sigma$, we use a specific variant of a Chernoff bound (in Section \ref{sec:m2disc}) that will depend only on $\sum_j \Delta_j^2$, where $\Delta_j = K(x,p_j) - K(x,q_j)$ for each edge $(p_j,q_j) \in M$.  
Note that $\sum_j |\Delta_j| \geq d_\chi(P,K)$ gives a bound on discrepancy, but analyzing this directly gives a $\poly(n)$ bound.  
Also, for a binary kernels $\sum_j \Delta_j^2 = \poly(n)$, but if the kernel slope is bounded then $\sum_j \Delta_j^2 \leq \sigma^2 \sum_j \|p_j - q_j\|^2$.    
Then in Section \ref{sec:match} we bound $\sum_j \|p_j - q_j\|^2=  O(1)$ within a constant radius ball, specifically the ball $B_x$ for which $K(x,\cdot) > 0$.  This follows (after some technical details) from a result of Bern and Eppstein~\cite{BE93}.  

Extending to $\Eu{G}$ requires multiple invocations (in Section \ref{sec:m2disc}) of the $\sum_j \Delta_j^2$ bound from Section \ref{sec:match}.  
Extending to $\Reals^d$ basically requires generalizing the matching result to depend on the sum of distances to the $d$th power (in Section \ref{sec:match}) and applying Jensen's inequality to relate $\sum_j \Delta_j^d$ to $\sum_j \Delta_j^2$ (in Section \ref{sec:m2disc}).  

\subsection{Motivation}

\paragraph{Near-linear sized $\eps$-samples.}
Only a limited family of range spaces are known to have $\eps$-samples with size near-linear in $1/\eps$.  Most reasonable range spaces in $\Reals^1$ admit an $\eps$-sample of size $1/\eps$ by just sorting the points and taking every $\eps |P|$th point in the sorted order.  However, near-linear results in higher dimensions are only known for range spaces defined by axis-aligned rectangles (and other variants defined by fixed, but non necessarily orthogonal axes) (c.f. \cite{STZ04,Phi08}).  
All results based on VC-dimension admit super-linear polynomials in $1/\eps$, with powers approaching $2$ as $d$ increases.  And random sampling bounds, of course, only provide $\eps$-samples of size $O(1/\eps^2)$.  

This polynomial distinction is quite important since for small $\eps$ (i.e. with $\eps = 0.001$, which is important for summarizing large datasets) then $1/\eps^2$ (in our example the size $1/\eps^2 = 1,000,000$) is so large it often defeats the purpose of summarization.  
Furthermore, most techniques other than random sampling (size $1/\eps^2$) are quite complicated and rarely implemented (many require Bansal's recent result~\cite{Ban10} or its simplification \cite{LM12}).  
One question explored in this paper is: what other families of ranges have $\eps$-samples with size near-linear in $1/\eps$?

This question has gotten far less attention recently than $\eps$-nets, a related and weaker summary.  In that context, a series of work~\cite{HW87,BEHW89,MSW90,KPW92,AES10,Ezr10} has shown that size bound of $O((1/\eps) \log (1/\eps))$ based on VC-dimension can be improved to $O((1/\eps) \log \log (1/\eps))$ or better in some cases.  Super-linear lower bounds are known as well~\cite{Alo10,PT10}.  
We believe the questions regarding $\eps$-samples are even more interesting because they can achieve polynomial improvements, rather than at best logarithmic improvements.  

\paragraph{$L_\infty$ kernel density estimates.}  
Much work (mainly in statistics) has studied the approximation properties of kernel density estimates; this grew out of the desire to remove the choice of where to put the breakpoints between bins in histograms.  
A large focus of this work has been in determining which kernels recover the original functions the best and how accurate the kernel density approximation is.  Highlights are the books of Silverman~\cite{Sil86} on $L_2$ error in $\Reals^1$, Scott~\cite{Sco92} one $L_2$ error in $\Reals^d$, and books by Devroye, Gy\"{o}rfi, and/or Lugosi~\cite{DG84,DL01,DGL96} on $L_1$ error.  We focus on $L_\infty$ error, which has not generally been studied.  

Typically the error is from two components: the error from use of a subset (between $\kde_P$ and $\kde_S$), and the error from convoluting the data with a kernel (between $P$ and $\kde_P$).  Here we ignore the second part, since it can be arbitrary large under the $L_\infty$ error. 
Typically in the first part, only random samples have been considered; we improve over these random sample bounds.  

Recently Chen, Welling, and Smola~\cite{CSW10} showed that for any positive definite kernel (including $\Eu{G}$, but not $\Eu{T}$ or $\Eu{B}$) a greedy process produces a subset of points $S \subset P$ such that $L_2(\kde_S, \kde_P) \leq \eps$ when $S$ is size $|S| = O(1/\eps)$.  This improves on standard results from random sampling theory that require $|S| = O(1/\eps^2)$ for such error guarantees.  
This result helped inspire us to seek a similar result under $L_\infty$ error.

\paragraph{Relationship with binary range spaces.}
Range spaces have typically required all ranges to be binary.  That is, each $p \in P$ is either completely in or completely not in each range $R \in \Eu{A}$.  Kernel range spaces were defined in a paper last year~\cite{JKPV11} (although similar concepts such as fat-shattering dimension~\cite{KS94} appear in the learning theory literature~\cite{BLW96,DGL96,ABCH97,Vap89}, their results were not as strong as \cite{JKPV11} requiring larger subsets $S$, and they focus on random sampling).  That paper showed that an $\eps$-sample for balls is also an $\eps$-sample for a large number of kernel range spaces.  An open question from that paper is whether that relationship goes the other direction: is an $\eps$-sample of a kernel range space also necessarily an $\eps$-sample for a binary range space?  
We answer this question in the negative; in particular, we show near-linear sized $\eps$-samples for kernel range spaces when it is known the corresponding binary range space must have super-linear size.  

This raises several questions: are these binary range spaces which require size super-linear in $1/\eps$ really necessary for downstream analysis?  Can we simplify many analyses by using kernels in place of binary ranges?  One possible target are the quite fascinating, but enormous bounds for bird-flocking~\cite{Cha10}.

\section{Preliminaries}
\label{sec:prelim}

For simplicity, we focus on only \emph{rotation-and-shift-invariant} kernels so $K(p_i,p_j) = k(\|p_i-p_j\|)$ can be written as a function of just the distance between its two arguments.  The rotation invariant constraint can easily be removed, but would complicate the technical presentation.  We also assume the kernels have been scaled so $K(p,p) = k(0) = 1$.  Section \ref{sec:ext} discusses removing this assumption.  

We generalize the family of kernels $\Eu{T}$ (see above) to $\Eu{S}_\sigma$ which we call \emph{$\sigma$-bounded}; they have slope bounded by a constant $\sigma > 0$ (that is for any $x,q,p \in \Reals^d$ then $|K(x,p) - K(x,q)|/\|q-p\| \leq \sigma$), and bounded domain $B_x = \{p \in \Reals^d \mid K(x,p) > 0\}$ defined for any $x \in \Reals^d$.  For a rotation-and-shift-invariant kernel, $B_x$ is a ball, and for simplicity we assume the space has been scaled so $B_x$ has radius $1$ (this is not critical, as we will see in Section \ref{sec:ext}, since as the radius increases, $\sigma$ decreases).  In addition to $\Eu{T}$ this class includes, for instance,  the Epanechnikov kernels $\Eu{E}$ so for $K(x,\cdot) \in \Eu{E}$ is defined $K(x,p) = \max \{0, 1-\|x-p\|^2\}$.  

We can also generalize $\Eu{G}$ (see above) to a family of exponential kernels such that $|k(z)| \leq \exp(-|\poly(z)|)$ and has bounded slope $\sigma$; this would also include, for instance, the Laplacian kernel. 
For simplicity, for the remainder of this work we focus technically on the Gaussian kernel.

Let $\rad(B)$ denote the radius of a ball $B$.  
The $d$-dimensional volume of a ball $B$ of radius $r$ is denoted $\vol_d(r) = (\pi^{d/2}/\Gamma(d/2+1)) r^d$ where $\Gamma(z)$ is the gamma function; note $\Gamma(z)$ is increasing with $z$ and when $z$ is a positive integer $\Gamma(z) = (z-1)!$.  
For balls of radius $1$ we set $\vol_d(1) = V_d$, a constant depending only on $d$.  In general $\vol_d(r) = V_d r^d$.  
Note that our results hold in $\Reals^d$ where $d$ is assumed constant, and $O(\cdot)$ notation will absorb terms dependent only on $d$.

\section{Min-Cost Matchings within Balls}
\label{sec:match}

Let $P \subset \Reals^d$ be a set of $2n$ points.  We say $M(P)$ (or $M$ when the choice of $P$ is clear) is a \emph{perfect matching} of $P$ if it defines a set of $n$ unordered pairs $(q,p)$ of points from $P$ such that every $p \in P$ is in exactly one pair.  We define \emph{cost of a perfect matching} in terms of the sum of distances:  
$
c(M) = \sum_{(q,p) \in M} ||q - p||.
$
The \emph{min-cost perfect matching} is the matching $M^* = \arg \min_{M} c(M)$. 

The proof of the main result is based on a lemma that relates the density of matchings in $M^*$ to the volume of a ball containing them.  For any ball $B \subset \Reals^d$, define the \emph{length} $\rho(B,M)$ of the matchings from $M$ within $B$ as:
\[
\rho(B,M) = \sum_{(q,p) \in M} 
\begin{cases}
\|q - p\|^d & \text{if } q,p \in B 
\\
\|q-p_{B}\|^d & \text{if } q \in B \text{ and } p \notin B \text{, where } p_{B} \text{ is the intersection of } \overline{q p} \text{ with } \partial B
\\ 
0 & \text{if } q,p \notin B
\end{cases}
\]

\begin{lemma}
There exists a constant $\phi_d$ that depends only on $d$, such that 
for any ball $B \subset \Reals^d$ the length 
$\rho(B,M^*) \leq (\phi_d/V_d) \vol_d(\rad(B)) = \phi_d \rad(B)^d$.
\label{lem:match-length}
\end{lemma}
\begin{proof}
A key observation is that we can associate each edge $(q, p) \in M^*$ with a shape that has $d$-dimensional volume proportional to $\|q - p\|^d$, and if these shapes for each edge do not overlap too much, then the sum of edge lengths to the $d$th power is at most proportional to the volume of the ball that contains all of the points.  

In fact, Bern and Eppstein~\cite{BE93} perform just such an analysis considering the minimum cost matching of points within $[0,1]^d$.  They bound the sum of the $d$th power of the edge lengths to be $O(1)$.  For an edge $(q, p)$, the shape they consider is a \emph{football}, which includes all points $y$ such that $q y p$ makes an angle of at most $170^\circ$; so it is thinner than the disk with diameter from $q$ to $p$, but still has $d$-dimensional volume $\Omega(\|q - p\|^d)$.    

To apply this result of Bern and Eppstein, we can consider a ball of radius $1/2$ that fits inside of $[0,1]^d$ by scaling down all lengths uniformly by $1/2\rad(B)$.  If we show the sum of $d$th power of edge lengths is $O(1)$ for this ball, and by scaling back we can achieve our more general result.  From now we will assume $\rad(B) = 1/2$.  Now the sum of $d$th power of edge lengths where both endpoints are within $B$ is at most $O(1)$ since these points are also within $[0,1]^d$.  

Handling edges of the second type where $p \notin B$ is a bit more difficult.  
Let $B$ be centered at a point $x$.  First, we can also handle all ``short'' edges $(q, p)$ where $\|x - p\| < 10$ since we can again just scale these points appropriately losing a constant factor, absorbed in $O(1)$.  

Next, we argue that no more than $360^{d-1}$ ``long'' edges $(q, p) \in M^*$ can exist with $q \in B$ and $\| p - x\| > 10$.  This implies that there are two such edges $(q, p)$ and $(q', p')$ where the angle between the vectors from $x$ to $p$ and from $x$ to $p'$ differ by at most $1$ degree.  We can now demonstrate that both of these edges cannot be in the minimum cost matching as follows.  

\noindent\includegraphics[width=\linewidth]{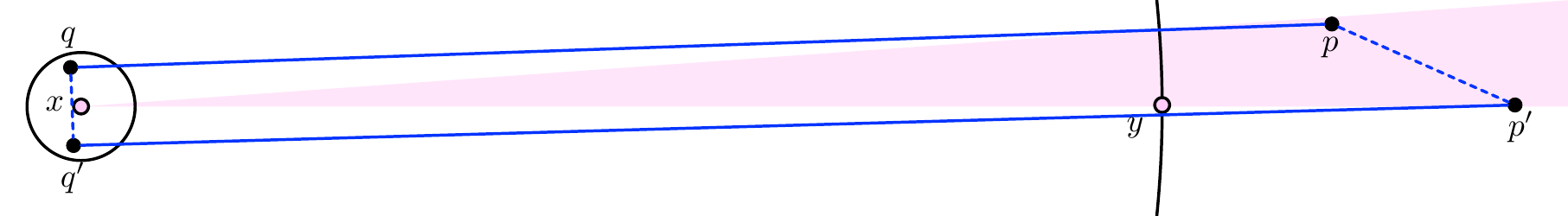}

Let $10 < \|x - p\| < \|x - p'\|$ and let $y$ be the point a distance $10$ from $x$ on the ray towards $p'$.  
We can now see the following simple geometric facts:
\textsf{(F1)} $\|q - q'\| < 1$, 
\textsf{(F2)} $\|p - y\| < \|p - x\|$, 
\textsf{(F3)} $\|p' - y\| + 10 = \|p' - x\|$, 
\textsf{(F4)} $\|p - y\| + \|p' - y\| \leq \|p - p'\|$, and
\textsf{(F5)} $\|p - x\| - 1/2 \leq \|p - q\|$ and $\|p'-x\| - 1/2 \leq \|p'-q'\|$.  
It follows that
\[
\|q - q'\| + \|p - p'\|
\leq
1 + \|p - y\| + \|p' - y\|
\leq
1 + \|p - x\| + \|p - x\| - 10
<
\|p - q\| + \|p' - q'\|.
\]
\vspace{-.3in}
\[
\text{via \textsf{(F1)} and \textsf{(F4)} \hspace{.6in} via \textsf{(F2)} and \textsf{(F3)} \hspace{1.1in} via \textsf{(F5)} \hspace{.3in}}
\]
Thus it would be lower cost to swap some pair of edges if there are more than $360^{d-1}$ of them.  Since each of these at most $360^{d-1}$ long second-type edges can contribute at most $1$ to $\rho(B,M^*)$ this handles the last remaining class of edges, and proves the theorem.  
\end{proof}

\section{Small Discrepancy for Kernel Range Spaces}
\label{sec:m2disc}

We construct a coloring $\chi : P \to \{-1,+1\}$ of $P$ by first constructing the minimum cost perfect matching $M^*$ of $P$, and then for each $(p_j,q_j) \in M^*$ we randomly color one of $\{p_j,q_j\}$ as $+1$ and the other $-1$.  

To bound the discrepancy of a single kernel $d_\chi(P,K)$ we consider a random variable $X_j = \chi(p_j) K(x,p_j) + \chi(q_j) K(x,q_j)$ for each pair $(p_j, q_j) \in M^*$, so $d_\chi(P,K) = |\sum_j X_j|$.  We also define a value $\Delta_j = 2| K(x,p_j) - K(x,q_j)|$ such that $X_j \in \{-\Delta_j/2, \Delta_j/2\}$.  
The key insight is that we can bound $\sum_j \Delta_j^2$ using the results from Section \ref{sec:match}.  
Then since each $X_j$ is an independent random variable and has $E[X_j] = 0$, we are able to apply a Chernoff bound on $d_{\chi}(P,K) = |\sum_j X_j|$ that says 
\begin{equation}
\label{eq:CH}
\prob{d_\chi (P,K) > \alpha} \leq 2\exp\left(\frac{-2\alpha^2}{\sum_j \Delta_j^2}\right).
\end{equation}
In $\Reals^2$, for $\sigma$-bounded kernels achieving a probabilistic discrepancy bound is quite straight-forward at this point (using Lemma \ref{lem:match-length}); and can be achieved for Gaussian kernels with a bit more work.  

However for points $P \subset \Reals^d$ for $d>2$ applying the above  bound is not as efficient since we only have a bound on $\sum_j \Delta_j^d$.  We can attain a weaker bound using the Jensen's inequality over at most $n$ terms
\begin{equation}
\left(\sum_j \frac{1}{n} \Delta_j^2\right)^{d/2} \leq \sum_j \frac{1}{n} \left(\Delta_j^2\right)^{d/2}
\;\; \text{ so we can state } \;\;
\sum_j \Delta_j^2 \leq n^{1-2/d} \left(\sum_j \Delta_j^d\right)^{2/d}.
\label{eq:Jensen}
\end{equation}

\paragraph{$\sigma$-bounded kernels.}  We start with the result for $\sigma$-bounded kernels using Lemma \ref{lem:match-length}.

\begin{lemma}
In $\Reals^d$, for any kernel $K \in \Eu{S}_\sigma$ we can construct a coloring $\chi$ such that $\Pr[d_\chi(P,K) > n^{1/2 - 1/d}\sigma (\phi_d)^{1/d} \sqrt{2\ln(2/\delta)}] \leq \delta$ for any $\delta > 0$.
\label{lem:discK}
\end{lemma}
\begin{proof}
Consider some $K(x,\cdot) \in \Eu{S}_\sigma$ and 
recall $B_x = \{y \in \Reals^2 \mid K(x,y) > 0\}$.  Note that 
\[
\sum_j \Delta_j^d 
= 
\sum_j 2^d(K(p_j,x) - K(q_j,x))^d 
\leq 
2^d \sigma^d \rho(B_x,M^*) 
\leq 
\sigma^d 2^d \phi_d \rad(B_x)^d 
<
\sigma^d 2^{d} \phi_d,
\]
where the first inequality follows by the slope of $K \in \Eu{T}$ being at most $\sigma$ and $K(p,x) = 0$ for $p \notin B$, since we can replace $p_j$ with $p_{j,B_x}$ when necessary since both have $K(x,\cdot) =0$, and the second inequality follows by Lemma \ref{lem:match-length}.  
Hence, by Jensen's inequality (i.e. (\ref{eq:Jensen})) $\sum_j \Delta_j^2 \leq n^{1-2/d} (\sigma^d 2^d \phi_d)^{2/d} = n^{1-2/d} \sigma^2 4 (\phi_d)^{2/d}$.  

We now study the random variable $d_\chi(P,K) = |\sum_i X_i|$ for a single $K \in \Eu{S}_\sigma$.  
Invoking (\ref{eq:CH}) we can bound
$\Pr[d_\chi(P,K) > \alpha] \leq 2 \exp(-\alpha^2/(n^{1-2/d} \sigma^2 2 (\phi_d)^{2/d}))$.
Setting $\alpha = n^{1/2 - 1/d} \sigma (\phi_d)^{1/d} \sqrt{2 \ln(2/\delta)}$ reveals $\Pr[d_\chi(P,K) > n^{1/2 - 1/d} \sigma (\phi_d)^{1/d} \sqrt{2\ln(2/\delta)}] \leq \delta$.  
\end{proof}

For $\Eu{T}$ and $\Eu{E}$ the bound on $\sigma$ is $1$ and $2$, respectively.  Also note that in $\Reals^2$ the expected discrepancy for any one kernel is independent of $n$.

\paragraph{Gaussian kernels.}  
Now we extend the above result for $\sigma$-bounded kernels to Gaussian kernels.  It requires a nested application of Lemma \ref{lem:match-length}.  
Let $z_i = 1/2^i$ and $B_i = \{p \in \Reals^d \mid K(x,p) \geq z_i\}$; let $A_i = B_i \setminus B_{i-1}$ be the annulus of points with $K(x,p) \in [z_i, z_{i-1})$.  For simplicity define $B_0$ as empty and $A_1 = B_1$.  
We can bound the slopes within each annulus $A_i$ as $\sigma_1 = 1 = \max \frac{d}{dy} (-k(y))$, and more specifically for $i \geq 2$ then $\sigma_i = \max_{y \in A_i} \frac{d}{dy}( - k(y)) = k(z_{i-1}) = 1/2^{i-1}$.  
Define $y_i = \sqrt{i \ln 2}$ so that $k(y_i) = z_i$.  

We would like to replicate Lemma \ref{lem:match-length} but to bound $\rho(B,M^*)$ for annuli $A_i$ instead of $B$.  However, this cannot be done directly because an annulus can become skinny, so its $d$-dimensional volume is not proportional to its width to the $d$th power.  This in turn allows for long edges within an annulus that do not correspond to a large amount of volume in the annulus.  

We can deal with this situation by noting that each edge either counts towards a large amount of volume towards the annulus directly, or could have already been charged to the ball contained in the annulus.  This will be sufficient for the analysis that follows for Gaussian kernels.  

Define $\rho(A_i,M)$ as follows for an annulus centered at point $x$.  For each edge $(p_j, q_j) \in M$ let $q_{j,i}$ be the point on $\overline{p_j q_j} \cap A_i$ furthest from $x$; if there is a tie, choose one arbitrarily.  Let $p_{j,i}$ be the point on $\overline{p_j q_j} \cap A_i$ closest to $x$, if there is a tie, choose the one closer to $q_{j,i}$.  
Then $\rho(A_i,M) = \sum_{(p_j,q_j) \in M} \|p_{j,i} - q_{j,i}\|^d$.  

\begin{lemma}
$\rho(A_i, M) \leq \rho(B_i,M) - \rho(B_{i-1},M)$ for $i \geq 1$.  
\label{lem:ann-bnd}
\end{lemma} 
\begin{proof}
For each edge $(p_j, q_j) \in M$ such that $\overline{p_j q_j}$ intersects $B_i$, has $\overline{p_i q_i} \cap B_i$ as either entirely, partially, or not at all in $A_i$.  Those that are entirely in or not at all in $A_i$ either contribute exclusively to $\rho(A_i,M)$ or $\rho(B_{i-1},M)$, respectively.  Those that are partially in $A_i$ contribute to $A_i$ and $B_{i-1}$ in a superadditive way towards $B_i$ leading to the claimed inequality.  
\end{proof}

Now notice that since $y_i = \sqrt{i \ln 2}$ then $\rho(B_i,M^*) \leq \phi_d (\ln 2)^{d/2} i^{d/2}$, by Lemma \ref{lem:match-length}.

\begin{lemma}
Let $\Psi(n,d,\delta) = O(n^{1/2 -1/d} \sqrt{\ln(1/\delta)})$.  
In $\Reals^d$, for any kernel $K \in \Eu{G}$ we can construct a coloring $\chi$ such that $\Pr[d_\chi(P,K) > \Psi(n,d,\delta)] \leq \delta$ for any $\delta > 0$.
\label{lem:disc-G}
\end{lemma}
\begin{proof}
Now to replicate the result in Lemma \ref{lem:discK} we need to show that $\sum_{(p_j, p'_j) \in M^*} (K(q_j) - K(p_j))^d < C$ for some constant $C$.  
For a kernel $K \in \Eu{G}$ centered at a point $x \in \Reals^d$ we can decompose its range of influence into a series of annuli $A_1, A_2, \ldots,$ as described above.  

Define $v_{i,j} = K(p_{j,i}) - K(q_{j,i})$ if both $p_{j,i}$ and $q_{j,i}$ exist, otherwise set $v_{i,j} = 0$.  Note that $v_{i,j} \leq \|p_{j,i} - q_{j,i}\|/ 2^{i-1}$ since the slope is bounded by $1/2^{i-1}$ in $A_i$.   

For any $(p_j, q_j) \in M^*$ we claim $K(q_j) - K(p_j) \leq \sum_i v_{i,j} \leq 3 \max_i v_{i,j}$.  The first inequality follows since the slope within annulus $A_i$ is bounded by $1/2^{i-1}$, so each $v_{i,j}$ corresponds to the change in value of $(p_j, q_j)$ with $A_i$.  
To see the second inequality, let $\ell$ be the smallest index such that $K(x,p_j) \geq z_\ell$, that is $A_\ell$ is the smallest annulus that intersects $\overline{q_j p_j}$.  We can see that $v_{\ell+1,j} \geq \sum_{i=\ell+2}^\infty v_{i,j}$ since $(y_i-y_{i-1})$ is a decreasing function of $i$ and $1/2^{i-1}$ is geometrically decreasing.  Thus the $\arg \max_i v_{i,j}$ is either $v_{\ell,j}$ or $v_{\ell+1,j}$, since $v_{i,j} = 0$ for $i < \ell$.  
If $v_{\ell,j} > v_{\ell+1,j}$ then $3 v_{\ell,j} > v_{\ell,j} + v_{\ell+1,j} + \sum_{i=\ell+2}^\infty v_{i,j} = \sum_{i=1}^\infty v_{i,j}$.  
If $v_{\ell,j} \leq v_{\ell+1,j}$ then $3 v_{\ell+1,j} > v_{\ell,j} + v_{\ell+1,j} + \sum_{i=\ell+2}^\infty v_{i,j} = \sum_{i=1}^\infty v_{i,j}$.  
Hence 
\[
(K(q_j) - K(p_j))^d 
\leq 
(\sum_i v_{i,j})^d 
<
(3 \max_i v_{i,j})^d 
\leq 
3^d \sum_i v_{i,j}^d.
\]

We can now argue that the sum over all $A_i$, that $\sum_i \sum_{(p_j,q_j) \in M^*} v_{i,j}^d$ is bounded.  By Lemma \ref{lem:ann-bnd}, and summing over all $(p_j, q_j)$ in a fixed $A_i$
\[
\sum_{(p_j,q_j) \in M^*} v_{i,j}^d 
= 
\sum_{(p_j,q_j) \in M^*} \frac{\|p_{j,i} - q_{j,i}\|^d}{(2^{i-1})^d}
=
\frac{\rho(A_i, M^*)}{2^{d(i-1)}}  
\leq 
\frac{\rho(B_i,M^*) - \rho(B_{i-1},M^*)}{2^{d(i-1)}}.
\]
Hence (using a reversal in the order of summation)
\begin{align*}
\sum_{(p_j,q_j) \in M^*} (K(q_j) - K(p_j))^d 
& \leq 
3^d \sum_{(p_j,q_j) \in M^*} \sum_{i=1}^\infty v_{i,j}^d  
= 
3^d \sum_{i=1}^\infty \sum_{(p_j,q_j) \in M^*} v_{i,j}^d 
\\ & \leq 
3^d \sum_{i=1}^\infty \frac{\rho(B_i,M^*) - \rho(B_{i-1},M^*)}{2^{d(i-1)}} 
=
3^d \sum_{i=1}^\infty \rho(B_i,M^*) \left(\frac{1}{2^{d(i-1)}} - \frac{1}{2^{di}} \right)
\\ & \leq
3^d \sum_{i=1}^\infty  \phi_d (\ln 2)^{d/2} i^{d/2} \left(\frac{2^d}{2^{di}}\right)
\leq 
6^d \phi_d,
\end{align*}
since for $i \geq 1$ and $d \geq 2$ the term $(i \ln 2)^{d/2} /2^{di} < 1/2^i$ is geometrically decreasing as a function of $i$.  

Finally we plug $\Delta_j = 2 |K(q_j) - K(p_j)|$ and using (\ref{eq:Jensen})
\[
\sum_j \Delta_j^2 
\leq 
n^{1-2/d} (\sum_j \Delta_j^d)^{2/d} 
\leq 
n^{1-2/d} (2^d \cdot 6^d \phi_d)^{2/d}
=
144 (\phi_d)^{2/d} n^{1-2/d} 
\] 
into a Chernoff bound  (\ref{eq:CH}) on $n$ pairs in a matching, as in the proof of Lemma \ref{lem:discK}, where $X_j$ represents the contribution to the discrepancy of pair $(q_j,p_j) \in M^*$ and $d_\chi(P,K) = |\sum_j X_j|$.   
Then 
\[
\Pr[d_\chi(P,K) > \alpha] 
\leq 
2 \exp\left(\frac{-2 \alpha^2}{ \sum_j \Delta_j^2} \right) 
\leq 
2 \exp\left(\frac{-\alpha^2 }{72 (\phi_d)^{2/d} n^{1-2/d}}\right).
\]  
Setting $\alpha = \Psi(n,d,\delta) = 6 \sqrt{2} (\phi_d)^{1/d} n^{1/2-1/d} \sqrt{\ln(2/\delta)}$ reveals $\Pr[d_\chi(P,K) > \Psi(n,d,\delta)] \leq \delta$.  
\end{proof}

Again, in $\Reals^2$ we can show the expected discrepancy for any one Gaussian kernel as independent of $n$.

\subsection{From a Single Kernel to a Range Space, and to $\eps$-Samples}
\label{sec:2fam}

The above theorems imply colorings with small discrepancy ($O(1)$  in $\Reals^2$)
for an arbitrary choice of $K \in \Eu{T}$ or $\Eu{E}$ or $\Eu{G}$, and give a randomized algorithm to construct such a coloring that does not use information about the kernel.  But this does not yet imply small discrepancy for all choices of $K \in \Eu{S}_\sigma$ or $\Eu{G}$ simultaneously.    To do so, we show we only need to consider a polynomial in $n$ number of kernels, and then show that the discrepancy is bounded for all of them.

Note for binary range spaces, this result is usually accomplished by deferring to VC-dimension $\nu$, where there are at most $n^\nu$ distinct subsets of points in the ground set that can be contained in a range.  Unlike for binary range spaces, this approach does not work for kernel range spaces since even if the same set $P_x \subset P$ have non-zero $K(x,p)$ for $p \in P_x$, their discrepancy $d_\chi(P,K_x)$ may change by recentering the kernel to $x'$ such that $P_x = P_{x'}$.  
Instead, we use the bounded slope (with respect to the size of the domain) of the kernel.  

For a kernel $K \in S_\sigma$ or $\Eu{G}$, let $B_{x,n} = \{p \in \Reals^d \mid K(x,p) > 1/n\}$.  
\begin{lemma}
For any $x \in \Reals^d$, $\vol_d(B_{x,n}) < V_d$ for $\Eu{S}_\sigma$ and $\vol_d(B_{x,n}) = V_d (\ln(2n))^{d/2}$ for $\Eu{G}$.  
\label{lem:vol-n}
\end{lemma}
\begin{proof}
For $\Eu{S}_\sigma$, clearly $B_{x,n} \subset B_x = \{p \in \Reals^d  \mid K(x,p) > 0\}$, and $\vol_d(B_x) = V_d$.  

For $\Eu{G}$, we have $k(z) = 1/n$ for $z = \sqrt{\ln(2n)}$, and a ball of radius $\sqrt{\ln(2n)}$ has $d$-dimensional volume $V_d (\sqrt{\ln(2n)})^d$.  
\end{proof}

\begin{theorem}
In $\Reals^d$, for $\Eu{K}$ as a family in $\Eu{S}_\sigma$ or $\Eu{G}$, and a value $\Psi(n,d,\delta) = O(n^{1/2 - 1/d} \sqrt{\log(n/\delta)})$, we can choose a coloring $\chi$ such that $\Pr[d_\chi(P,\Eu{K}) > \Psi(n,d,\delta)] \leq \delta$ for any $\delta > 0$.
\label{thm:poly-num-kern}
\end{theorem}
\begin{proof}
Each $p \in P$ corresponds to a ball $B_{p,n}$ where $K(p,\cdot) > 1/n$.  Let $U = \bigcup_{p \in P} B_{p,n}$.  For any $q \notin U$, then $\sum_{p\in P} K(p,q) \leq 1$, and thus $d_\chi(P,K) \leq 1$; we can ignore these points.  

Also, by Lemma \ref{lem:vol-n}, the $d$-dimensional volume of $U$ is at most $V_d n$ for $\Eu{S}_\sigma$ and at most $V_d n (\ln(2n))^{d/2}$ for $\Eu{G}$.  
We can then cover $U$ with a net $N_\tau$ such that for each $x \in U$, there exists some $q \in N_\tau$ such that $\|x - q\| \leq \tau$.  Based on the $d$-dimensional volume of $U$, there exists such a net of size $|N_\tau| \leq O(\tau^d n)$ for $\Eu{S}_\sigma$ and $|N_\tau| \leq O(\tau^d n (\ln(2d))^{d/2})$ for $\Eu{G}$.  

The maximum slope of a kernel $K \in \Eu{G}$ is $\sigma = 1$.   
Then, if for all $q \in N_\tau$ we have $d_\chi(P,K_q) \leq D$ (where $D \geq 1$), then any point $x \in \Reals^d$ has discrepancy at most $d_\chi(P,K_x) \leq D + \tau n \sigma$.  Recall for any $x \notin U$, $d_\chi(P,K_x) < 1$.  Then the $\tau n \sigma$ term follows since by properties of the net we can compare discrepancy to a kernel $K_q$ shifted by at most $\tau$, and thus this affects the kernel values on $n$ points each by at most $\sigma \tau$.  Setting, $\tau = 1/n\sigma$ it follows that $d_\chi(P,K_x) \leq D + 1$ for any $x \in \Reals^d$.  
Thus for this condition to hold, it suffices for $|N_q| = O(n^{d+1} \sigma^d)$ for $\Eu{S}_\sigma$ and $|N_q| = O(n^{d+1} (\ln(2n))^{d/2})$ for $\Eu{G}$.  

Setting the probability of failure in Lemma \ref{lem:discK} to $\delta'$ for each such kernel to $\delta' = \Omega(\delta/|N_\tau|)$ implies that for some value $\Psi(n,d,\delta) = n^{1/2 - 1/d} \sigma (\phi_d)^{1/d} \sqrt{2 \ln (2/\delta')} + 1 = O(n^{1/2 - 1/d} \sqrt{\log(n /\delta)})$, for $\Eu{K}$ as $\Eu{S}_\sigma$ or $\Eu{G}$, the $\Pr[d_\chi(P,\Eu{K}) > \Psi(n,d,\delta)] \leq \delta$.  
\end{proof}

To transform this discrepancy algorithm into one for $\eps$-samples, we need to repeat it successfully $O(\log n)$ times.  Thus to achieve success probability $\phi$ we can set $\delta = \phi / \log n$ above, and get a discrepancy of at most $\sqrt{2 \log (\frac{n}{\phi} \log n)}$ with probability $1-\phi$.
Now this and \cite{Phi08} implies Theorem \ref{thm:d-main}.  
We can state the corollary below using Varadarajan's $O(n^{1.5} \log^5 n)$ time algorithm~\cite{Var98} for computing the min-cost matching in $\Reals^2$; the $\log$ factor in the runtime can be improved using better SSPD constructions~\cite{AH12} to $O(n^{1.5} \log^2 n)$ expected time or $O(n^{1.5} \log^3 n)$ deterministic time.  

\begin{corollary}
For any point set $P \subset \Reals^2$, for any class $\Eu{K}$ of $\sigma$-bounded kernels or Gaussian kernels, in $O((n/\sqrt{\eps}) \log^2 (1/\eps))$ expected time (or in $O((n/\sqrt{\eps}) \log^3 (1/\eps))$ deterministic time) we can create an $\eps$-sample of size $O((1/\eps) \sqrt{\log(1/\eps\phi)})$ for $(P,\Eu{K})$, with probability at least $1-\phi$. 
\end{corollary}

\section{Extensions}
\label{sec:ext}

\paragraph{Bandwidth scaling.}
A common topic in kernel density estimates is fixing the integral under the kernel (usually at $1$) and the ``bandwidth'' $w$ is scaled.  For a shift-and-rotation-invariant kernel, where the default kernel $k$ has bandwidth $1$, a kernel with bandwidth $w$ is written and defined $k_w(z) = (1/w^d) k(z/w)$.  

Our results do not hold for arbitrarily small bandwidths, because then $k_w(0) = 1/w^d$ becomes arbitrarily large as $w$ shrinks; see Appendix \ref{sec:lb}.  
However, fix $W$ as some small constant, and consider all kernels $\Eu{K}_W$ extended from $\Eu{K}$ to allow any bandwidth $w \geq W$.  We can construct an $\eps$-sample of size $O((1/\eps)^{1/2-1/d} \sqrt{\log(1/\eps\delta)})$ for $(P,\Eu{K}_W)$.  
The observation is that increasing $w$ to $w'$ where $1 < w'/w < \eta$ affects $d_\chi(P,k_w)$ by at most $O(n \eta^{d+1})$ since $\sigma$ and $W$ are assumed constant.  Thus we can expand $N_\tau$ by a factor of $O(n^{d+1} \log n)$; for each $x \in N_\tau$ consider $O(n^{d+1}\log n)$ different bandwidths increasing geometrically by $(1+\sigma/n^{d+1})$ until we reach $n^{1/(d+1)}$.  For $w > n^{1/(d+1)}$ the contribution of each of $n$ points to $d_\chi(P,K_x)$ is at most $1/n$, so the discrepancy at any $x$ is at most $1$.  
$N_\tau$ is still $\poly(n)$, and hence Theorem \ref{thm:poly-num-kern} still holds for $\Eu{K}_W$ in place of some $\Eu{K} \in \{\Eu{S}_\sigma, \Eu{G}\}$.  

Another perspective in bandwidth scaling is how it affects the constants in the bound directly.  The main discrepancy results (Theorem \ref{thm:poly-num-kern}) would have a discrepancy bound $O(r \sigma n^{1/2 -1/d} \sqrt{d \log (n/\delta)})$ where $r$ is the radius of a region of the kernel that does not have a sufficiently small $K(p,\cdot)$ value, and $\sigma$ is the maximum slope.  Note that $k(0) \leq r \sigma$.  
Assume $\sigma$ and $r$ represent the appropriate values when $w=1$, then as we change $w$ and fix the integral under $K_w(x,\cdot)$ at $1$, we have $\sigma_w = \sigma/w^d$ and $r_w = r w$ so $r_w \sigma_w = r \sigma / w^{d-1}$.  
Thus as the bandwidth $w$ decreases, the discrepancy bound increases at a rate inversely proportional to $1/w^{(d-1)}$.  

Alternatively we can fix $k(0) = 1$ and scale $k_s(z) = k(z/s)$.  This trades off $r_s = rs$ and $\sigma_s = \sigma/s$, so $r_s \sigma_s = r \sigma$ is fixed for any value $s$.  Again, we can increase $N_\tau$ by a factor $O(n \log n)$ and cover all values $s$.  

\paragraph{Lower bounds.}
We present several straight-forward lower bounds in Appendix \ref{sec:lb}.  In brief we show:
\begin{itemize}\denselist
\item random sampling cannot do better on kernel range spaces than on binary range spaces (size $O(1/\eps^2)$);
\item  $\eps$-samples for kernels requires $\Omega(1/\eps)$ points;
\item $\eps$-samples for $(P,\Eu{B})$ require size $\Omega(1/\eps^{2d/(d+1)})$ in $\Reals^d$; and
\item kernels with $k(0)$ arbitrary large can have unbounded discrepancy.  
\end{itemize}

\subsection{Future Directions}
We believe we should be able to make this algorithm deterministic using iterative reweighing of the $\poly(n)$ kernels in $N_\tau$.  We suspect in $\Reals^2$ this would yield discrepancy $O(\log n)$ and an $\eps$-sample of size $O((1/\eps) \log (1/\eps))$.  

We are not sure that the bounds in $\Reals^d$ for $d>2$ require super-linear in $1/\eps$ size $\eps$-samples.  The polynomial increase is purely from Jensen's inequality.  Perhaps a matching that directly minimizes sum of lengths to the $d$th or $(d/2)$th power will yield better results, but these are harder to analyze.  
Furthermore, we have not attempted to optimize for, or solve for any constants.  A version of Bern and Eppstein's result~\cite{BE93} with explicit constants would be of some interest.  

We provide results for classes of $\sigma$-bounded kernels $\Eu{S}_\sigma$, Gaussian kernels $\Eu{G}$, and ball kernels $\Eu{B}$.  In fact this covers all classes of kernels described on the Wikipedia page on statistical kernels:
\url{http://en.wikipedia.org/wiki/Kernel_(statistics)}. 
However, we believe we should in principal be able to extend our results to all shift-invariant kernels with bounded slope.  
This can include kernels which may be negative (such as the \emph{sinc} or \emph{trapezoidal} kernels~\cite{DG84,DL01}) and have nice $L_2$ $\kde$ approximation properties.  These sometimes-negative kernels cannot use the $\eps$-sample result from~\cite{JKPV11} because their super-level sets have unbounded VC-dimension since $k(z)=0$ for infinitely many disjoint values of $z$.  

Edmond's min-cost matching algorithm~\cite{Edm65} runs in $O(n^3)$ time in $\Reals^d$ and Varadarajan's improvement~\cite{Var98} runs in $O(n^{1.5} \log^5 n)$ in $\Reals^2$ (and can be improved to $O(n^{1.5} \log^2 n)$ \cite{AH12}).  However, $\eps$-approximate algorithms~\cite{VA99} run in $O((n/\eps^3) \log^6 n)$ in $\Reals^2$.  As this result governs the runtime for computing a coloring, and hence an $\eps$-sample, it would be interesting to see if even a constant-factor approximation could attain the same asymptotic discrepancy results; this would imply and $\eps$-sample algorithm that ran in time $O(n \cdot \poly \log(1/\eps))$ for $1/\eps < n$.  

\subsection*{Acknowledgements}
I am deeply thankful to Joel Spencer for discussions on this problem, and for leading me to other results that helped resolved some early technical issues.  He also provided through personal communication a proof (I believe discussed with Prasad Tetali) that if the matching directly minimizes the sum of squared lengths, then same results hold.  
Finally, I thank David Gamarnik for communicating the Bern-Eppstein result \cite{BE93}.

\newpage

\bibliographystyle{plain}
\bibliography{discrepancy}

\begin{thebibliography}{10}

\bibitem{AH12}
Mohammad~A. Abam and Sariel Har-Peled.
\newblock New constructions of sspds and their applications.
\newblock {\em Computational Geometry: Theory and Applications}, 45:200--214,
  2012.

\bibitem{AAB08}
H\"useyin Akcan, Alex Astashyn, and Herv\'{e} Br\"{o}nnimann.
\newblock Deterministic algorithms for sampling count data.
\newblock {\em Data \& Knowledge Engineering}, 64:405--418, 2008.

\bibitem{Alo10}
Noga Alon.
\newblock A non-linear lower bound for planar epsilon-nets.
\newblock In {\em Proceedings 51st Annual IEEE Symposium on Foundations of
  Computer Science}, 2010.

\bibitem{ABCH97}
Noga Alon, Shai Ben-David, Nocol\`o Cesa-Bianchi, and David Haussler.
\newblock Scale-sensitive dimensions, uniform convergence, and learnability.
\newblock {\em Journal of ACM}, 44:615--631, 1997.

\bibitem{AES10}
Boris Aronov, Esther Ezra, and Micha Sharir.
\newblock Small-size epsilon-nets for axis-parallel rectangles and boxes.
\newblock {\em SIAM Journal of Computing}, 39:3248--3282, 2010.

\bibitem{Ban10}
Nikhil Bansal.
\newblock Constructive algorithms for discrepancy minimization.
\newblock In {\em Proceedings 51st Annual IEEE Symposium on Foundations of
  Computer Science}, pages 407--414, 2010.

\bibitem{BLW96}
Peter~L. Bartlett, Philip~M. Long, and Robert~C. Williamson.
\newblock Fat-shattering and the learnability of real-valued functions.
\newblock {\em Journal of Computer and System Sciences}, 52(3):434--452, 1996.

\bibitem{BE93}
Marshall Bern and David Eppstein.
\newblock Worst-case bounds for subadditive geometric graphs.
\newblock {\em Proceedings 9th ACM Symposium on Computational Geometry}, 1993.

\bibitem{BEHW89}
Anselm Blumer, A.~Ehrenfeucht, David Haussler, and Manfred~K. Warmuth.
\newblock Learnability and the {V}apnik-{C}hervonenkis dimension.
\newblock {\em Journal of ACM}, 36:929--965, 1989.

\bibitem{Cha01}
Bernard Chazelle.
\newblock {\em The Discrepancy Method}.
\newblock Cambridge, 2000.

\bibitem{Cha10}
Bernard Chazelle.
\newblock The convergence of bird flocking.
\newblock In {\em Proceedings 26th Annual Symposium on Computational Geometry},
  2010.

\bibitem{CM96}
Bernard Chazelle and Jiri Matousek.
\newblock On linear-time deterministic algorithms for optimization problems in
  fixed dimensions.
\newblock {\em J. Algorithms}, 21:579--597, 1996.

\bibitem{CSW10}
Yutian Chen, Max Welling, and Alex Smola.
\newblock Super-samples from kernel hearding.
\newblock In {\em Conference on Uncertainty in Artificial Intellegence}, 2010.

\bibitem{CCD11}
Edith Cohen, Graham Cormode, and Nick Duffield.
\newblock Structure-aware sampling: Flexible and accurate summarization.
\newblock In {\em Proceedings 37th International Conference on Very Large Data
  Bases}, 2011.

\bibitem{DG84}
Luc Devroye and L\'{a}szl\'{o} Gy\"{o}rfi.
\newblock {\em Nonparametric Density Estimation: The $L_1$ View}.
\newblock Wiley, 1984.

\bibitem{DGL96}
Luc Devroye, L\'aszl\'o Gy\"orfi, and G\'abor Lugosi.
\newblock {\em A Probabilistic Theory of Pattern Recognition}.
\newblock Springer-Verlag, 1996.

\bibitem{DL01}
Luc Devroye and G\'{a}bor Lugosi.
\newblock {\em Combinatorial Methods in Density Estimation}.
\newblock Springer-Verlag, 2001.

\bibitem{Edm65}
Jack Edmonds.
\newblock Paths, trees, and flowers.
\newblock {\em Canadian Journal of Mathematics}, 17:449--467, 1965.

\bibitem{Ezr10}
Esther Ezra.
\newblock A note about weak epsilon-nets for axis-parallel boxes in d-space.
\newblock {\em Information Processing Letters}, 110:18--19, 2010.

\bibitem{HW87}
David Haussler and Emo Welzl.
\newblock epsilon-nets and simplex range queries.
\newblock {\em Disc. \& Comp. Geom.}, 2:127--151, 1987.

\bibitem{JKPV11}
Sarang Joshi, Raj~Varma Kommaraju, Jeff~M. Phillips, and Suresh
  Venkatasubramanian.
\newblock Comparing distributions and shapes using the kernel distance.
\newblock In {\em 27th Annual Symposium on Computational Geometry}, 2011.

\bibitem{KS94}
Michael Kerns and Robert~E. Shapire.
\newblock Efficient distribution-free learning of probabilistic concepts.
\newblock {\em Journal of Computer and System Sciences}, 48:464--497, 1994.

\bibitem{KPW92}
J.~Koml\'os, J\'{a}nos Pach, and Gerhard Woeginger.
\newblock Almost tight bounds for epsilon nets.
\newblock {\em Discrete and Computational Geometry}, 7:163--173, 1992.

\bibitem{LLS01}
Yi~Li, Philip~M. Long, and Aravind Srinivasan.
\newblock Improved bounds on the samples complexity of learning.
\newblock {\em J. Comp. and Sys. Sci.}, 62:516--527, 2001.

\bibitem{LM12}
Shachar Lovett and Raghu Meka.
\newblock Constructive discrepancy minimization by walking on the edges.
\newblock Technical report, arXiv:1203.5747, 2012.

\bibitem{Mat99}
Jiri Matousek.
\newblock {\em Geometric Discrepancy; An Illustrated Guide}.
\newblock Springer, 1999.

\bibitem{MSW90}
Jiri Matou\v{s}ek, Raimund Seidel, and Emo Welzl.
\newblock How to net a lot with a little: Small $\eps$-nets for disks and
  halfspaces.
\newblock In {\em Proceedings 6th Annual ACM Symposium on Computational
  Geometry}, 1990.

\bibitem{Mat08}
Jir\v{r}i Matou\v{s}ek and Jan Vondrak.
\newblock The probabilistic method (lecture notes).
\newblock \url{http://kam.mff.cuni.cz/~matousek/prob-ln.ps.gz}, March 2008.

\bibitem{PT10}
J\'{a}nos Pach and G\'abor Tardos.
\newblock Tight lower bounds for the size of epsilon-nets.
\newblock In {\em Proceedings 27th Annual Symposium on Computational Geometry},
  2010.

\bibitem{Phi08}
Jeff~M. Phillips.
\newblock Algorithms for $\eps$-approximations of terrains.
\newblock In {\em ICALP}, 2008.

\bibitem{RACZU11}
Matteo Riondato, Mert Akdere, Ugur Cetintemel, Stanley~B. Zdonik, and Eli
  Upfal.
\newblock The {VC}-dimension of queries and selectivity estimation through
  sampling.
\newblock Technical report, arXiv:1101.5805, 2011.

\bibitem{Sco92}
David~W. Scott.
\newblock {\em Multivariate Density Estimation: Theory, Practice, and
  Visualization}.
\newblock Wiley, 1992.

\bibitem{Sil86}
Bernard~W. Silverman.
\newblock {\em Density Estimation for Statistics and Data Analysis}.
\newblock Chapman \& Hall/CRC, 1986.

\bibitem{STZ04}
S.~Suri, C.~Toth, and Y.~Zhou.
\newblock Range counting over multidimensional data streams.
\newblock {\em Discrete and Computational Geometry}, 36(4):633--655, 2006.

\bibitem{Vap89}
Vladimir Vapnik.
\newblock Inductive principles of the search for emperical dependencies.
\newblock In {\em Proceedings of the Second Annual Workshop on Computational
  Learning Theory}, pages 3--21, 1989.

\bibitem{VC71}
Vladimir Vapnik and Alexey Chervonenkis.
\newblock On the uniform convergence of relative frequencies of events to their
  probabilities.
\newblock {\em The. of Prob. App.}, 16:264--280, 1971.

\bibitem{Var98}
Kasturi~R. Varadarajan.
\newblock A divide-and-conquer algorithm for min-cost perfect matching in the
  plane.
\newblock In {\em Proceedings 39th IEEE Symposium on Foundations of Computer
  Science}, 1998.

\bibitem{VA99}
Kasturi~R. Varadarajan and Pankaj~K. Agarwal.
\newblock Approximation algorithms for bipartite and non-bipartite matching in
  the plane.
\newblock In {\em Proceedings 10th ACM-SIAM Symposium on Discrete Algorithms},
  1999.

\end{thebibliography}

\appendix

\section{Lower Bounds}
\label{sec:lb}

We state here a few straight-forward lower bounds.  These results are not difficult, and are often left implied, but we provide them here for completeness.  

\paragraph{Random sampling.}
We can show that in $\Reals^d$ for a constant $d$, a random sample of size $O((1/\eps^2)\log (1/\delta))$ 
is an $\eps$-sample for all shift-invariant, non-negative, non-increasing kernels (via the improved bound~\cite{LLS01} on VC-dimension-bounded range spaces~\cite{VC71} since they describe super-level sets of these kernels using~\cite{JKPV11}).  We can also show that we cannot do any better, even for one kernel $K(x,\cdot)$. 
Consider a set $P$ of $n$ points where $n/2$ points $P_a$ are located at location $a$, and $n/2$ points $P_b$ are located at location $b$.  Let $\|a - b\|$ be large enough that $K(a,b) < 1/n^2$ (for instance for $K \in \Eu{T}$ let $K(a,b) = 0$).  Then for an $\eps$-sample $S \subset P$ of size $k$, we require that $S_a \subset P_a$ has size $k_a$ such that $|k_a - k/2| < \eps k$.  
The probability (via Proposition 7.3.2 \cite{Mat08} for $\eps \leq 1/8$)
\[
\prob{|P_a| \geq k/2 + \eps k)} 
\geq
\frac{1}{15} \exp(-16 (\eps k)^2/k)
\geq
\frac{1}{15} \exp(-16 \eps^2 k) 
\geq 
\delta.
\]
Solving for $k$ reveals $k \leq (1/16 \eps^2) \ln(1/15\delta)$.  Thus if $k = o((1/\eps^2)\log(1/\delta))$ the random sample will have more than $\delta$ probability of having more than $\eps$-error.  
  
\paragraph{Requirement of $1/\eps$ points for $\eps$-samples.}
Consider a set $P$ where there are only $t = \lceil 1/\eps \rceil  - 1$ distinct locations of points, each containing $|P|/t$ points from $P$.  Then we can consider a range for each of the distinct points that contains only those points, or a kernel $K(x_i,0)$ for $i \in [t]$ that is non-zero  (or very small) for only the points in the $i$th location.  Thus if any distinct location in $P$ is not represented in the $\eps$-sample $S$, then the query on the range/kernel registering just those points would have error greater than $\eps$.  Thus an $\eps$-sample must have size at least $\lceil 1/\eps \rceil -1$.  
   
\paragraph{Minimum $\eps$-sample for balls.}
We can achieve a lower bound for the size of $\eps$-samples in terms of $\eps$ for any range space $(P,\Eu{A})$ by leveraging known lower bounds from discrepancy.  For instance, for $(P,\Eu{B})$ it is known (see \cite{Mat99} for many such results) that there exists points sets $P \subset \Reals^d$ of size $n$ such that $d(P,\Eu{B}) > \Omega(n^{1/2 - 1/2d})$.  
To translate this to a lower bound for the size of an $\eps$-sample $S$ on $(P,\Eu{B})$ we follow a technique outlined in Lemma 1.6 \cite{Mat99}.  
An $\eps$-sample of size $n/2$ implies for an absolute constant $C$ that
\[
C n^{1/2-1/2d} 
\leq 
d(P,\Eu{B}) 
\leq
\eps n.
\]
Hence, solving for $n/2$ (the size of our $\eps$-sample) in terms of $\eps$ reveals that $|S| = n/2 \geq (1/2) (C/\eps)^{2d/(d+1)}$.  Hence for $\eps$ small enough so $|S| \geq n/2$ reveals that the size of the $\eps$-sample for $(P,\Eu{B})$ is $\Omega(1/\eps^{2d/(d+1)})$.

\paragraph{Delta kernels.}
We make the assumption that $K(x,x) = 1$, if instead $K(x,x) = \eta$ and $\eta$ is arbitrarily large, then the discrepancy $d(P,K) = \Omega(\eta)$.  We can place one point $p \in P$ far enough from all other points $p' \in P$ so $K(p,p')$ is arbitrarily small.  Now $d_\chi(P,K(p,\cdot))$ approaches $\eta$, no matter whether $\chi(p) = +1$ or $\chi(p) = -1$.

\end{document}